\documentclass[envcountsame, envcountsect, runningheads]{llncs}

\usepackage[T1]{fontenc}
\usepackage{color}

\usepackage{graphicx}
\usepackage{subfigure}
\usepackage[utf8]{inputenc}
\usepackage{amsmath}
\usepackage{amssymb,cite}
\usepackage{hyperref}

\spnewtheorem{obsv}[theorem]{Observation}{\bfseries}{\itshape}

\newcommand{\overbar}[1]{\mkern 1.5mu\overline{\mkern-1.5mu#1\mkern-1.5mu}\mkern 1.5mu}

\title{Playing Snake on a Graph} 

\author{Denise Graafsma \and Bodo Manthey \and Alexander Skopalik}

\institute{University of Twente\\ Mathematics of Operations Research \\ Enschede, The Netherlands}

\begin{document}
\maketitle

\begin{abstract}
Snake is a classic computer game, which has been around for decades. Based on this game, we study the game of Snake on arbitrary undirected graphs. A snake forms a simple path that has to move to an apple while avoiding colliding with itself. When the snake reaches the apple, it grows longer, and a new apple appears. A graph on which the snake has a strategy to keep eating apples until it covers all the vertices of the graph is called snake-winnable.

We prove that determining whether a graph is snake-winnable is $\mathsf{NP}$-hard, even when restricted to grid graphs. We fully characterize snake-winnable graphs for odd-sized bipartite graphs and graphs with vertex-connectivity 1. While Hamiltonian graphs are always snake-winnable, we show that non-Ham\-il\-tonian snake-winnable graphs have a girth of at most 6 and that this bound is tight. 
\end{abstract}

\section{Introduction}\label{sec: intro}
In the game of Snake, players control a snake and have to guide it to apples that appear on the screen. With each apple consumed, the snake grows longer. The challenge is to grow the snake as long as possible while avoiding collisions with its own body. Based on this game, we study the game of Snake on a graph.

\subsection{Related work}
Over the years, many puzzles and games have been studied through a mathematical lens. Demaine and Hearn provide an extensive overview of the area of combinatorial games \cite{overview_games}, and Kendall et. al. published a survey of $\mathsf{NP}$-complete puzzles \cite{kendall2008survey}. A number of combinatorial games are played on a graph and involve two players with different roles. In pursuit-evasion games, such as Cops and Robbers \cite{Bonato2011,NOWAKOWSKI1983235,AIGNER19841,BONATO20095588}, one plays the pursuer, while the other tries to avoid getting caught. In Maker-Breaker games \cite{Hefetz2014,DUCHENE2025502,CHVATAL1978221}, the players alternately claim previously unclaimed edges. The Maker tries to possess a winning set, while the Breaker wants to prevent this. An example of a Maker-Breaker game is the Hamiltonicity game, where the Maker tries to possess a set of edges that form a Hamiltonian cycle on a complete graph. For these types of games, one of the most important questions is: given a graph, does one of the two players have a winning strategy. Nowakowski and Winkler gave a characterization for the graphs with a winning cop strategy \cite{NOWAKOWSKI1983235}. For the Hamiltonicity game, Chvátal and Erd\H{o}s proved that for large enough complete graphs, the Maker has a winning strategy \cite{CHVATAL1978221}. While the game of Snake on a graph can neither be described as a pursuit-evasion game, nor as a Maker-Breaker game, it does have two players with different roles: the snake who wants to grow as long as possible and the apple placer who tries to prevent this.

The work by De Biasi and Ophelders \cite{Biasi2018} on the Nibbler food collection problem is inspired by a different variant of Snake known as Nibbler. This problem asks the following: given a graph, the food locations, growth rate, and starting position of the snake, can the snake collect all the food? The growth rate indicates how much the snake grows each time it eats a piece of food. The Nibbler food collection problem is $\mathsf{NP}$-hard, even when restricted to solid grid graphs. Furthermore, if the growth rate is at least 2, then it is also $\mathsf{NP}$-hard on rectangular grid graphs.

\subsection{Our contributions}
While the work by De Biasi and Ophelders is primarily inspired by the game Nibbler, the version of Snake we consider is closer to the Nokia version. Instead of having all the apples on the graph from the beginning, only one apple is present at a time. A new apple is placed only when the previous one is consumed, meaning the snake does not know the locations of future apples. Hence, it will have to adjust its strategy according to where the next apple appears. In contrast to the game described by De Biasi and Ophelders, which is framed as a motion planning problem, our game can be viewed as a two-player game where the apple placer acts as an adversary to the snake. We also generalize the game to be played on any connected graph, rather than strictly adhering to the original game of Snake and only considering grid graphs.

For Snake on a graph, we show that determining whether the snake has a winning strategy is $\mathsf{NP}$-hard, even when restricted to grid graphs (Section~\ref{sec: complexity}). We also prove that the snake can never win on non-Hamiltonian graphs with a girth greater than 6 (Section~\ref{sec: girth}). Due to space constraints, some proofs have been omitted and can be found in the appendix.

\section{The game of Snake on a graph}\label{sec: game_desc}
In this section, we formally define the game of Snake on a graph and introduce some notation. The game is played on a simple graph~$G=(V,E)$ that is connected and has~$|V|\geq~3$. From now on, we assume all our graphs have this property.

 \subsection{Snake position}
  During the game, the snake always occupies an ordered set of vertices, which must form a simple path. We define the \textit{length} of the snake as the number of vertices on this path. Since the length of a path is defined as the number of edges it contains, the length of the snake is always one more than the length of the path it forms. 
  
  Let~$\ell$ be the current length of the snake. We denote the position of the snake by~$S=(s_1,\ldots, s_{\ell})$. We refer to~$s_1$ as the head, and~$s_{\ell}$ as the tail of the snake. By~$\overline{S}$, we denote all vertices in~$V$ that are not on~$S$. We refer to~$\overline{S}$ as the \textit{unoccupied set}.

  In some cases, we index the snake's position by time to better describe the snake's movement. By~$S^t=(s_1^t,\ldots, s_{\ell}^t)$, we denote the position of the snake at time~$t$.

 \subsection{Snake movement}
 Let~$S^t=(s_1^t,\ldots, s_{\ell}^t)$ be the current position of the snake. For a vertex~$v\in V$, the neighbor set~$N(v)$ denotes the set of vertices in~$V$ that are adjacent to~$v$. The head of the snake must move to a vertex in~$N(s_1^t)$. Suppose the head moves from~$s_1^t$ to some vertex~$v\in N(s_1^t)$. Then the next position of the snake becomes~$S^{t+1}=(v,s_1^t,\ldots,s_{\ell-1}^t)$. In other words, we add~$v$ to the beginning of the path and remove~$s_{\ell}^t$ from the end. 
 
 By our rules,~$S^{t+1}$ should still form a simple path. It follows that the head can move to any adjacent vertex that is either unoccupied or the current tail vertex, as depicted in Figure~\ref{fig: game: legal moves grid}. More formally, the head must move to some vertex in~$N(s_1^t)\cap\left(\overline{S^t}\cup\{s_\ell^t \}\right)$. In Figure~\ref{fig: game: illegal move grid}, we can see that if the snake moves to an occupied vertex that is not the tail, then the snake will no longer form a simple path. Hence, we forbid this type of movement.

 \begin{figure}[t]
     \centering
     \subfigure[The snake can move to an unoccupied vertex or the tail vertex.]{\includegraphics[width = 0.4\textwidth]{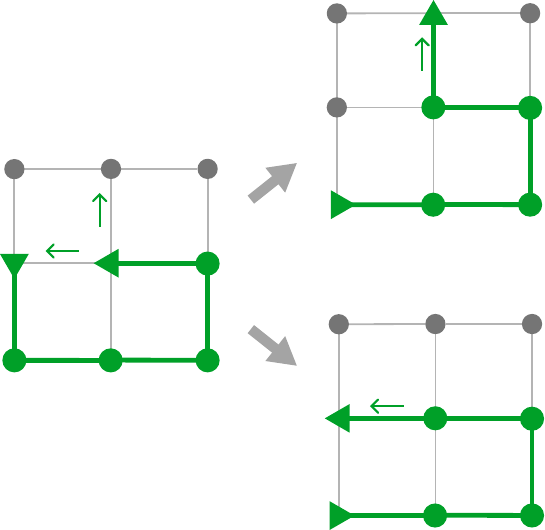}\label{fig: game: legal moves grid}}\hspace{1cm}
     \subfigure[After an illegal move, the snake no longer forms a simple path.]{\includegraphics[width = 0.4\textwidth]{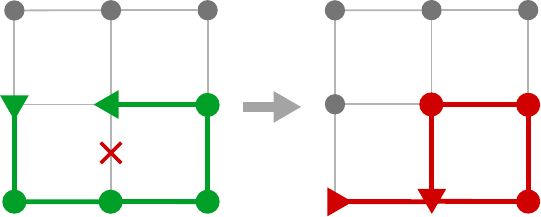}\label{fig: game: illegal move grid}}%
    \caption{The rules for the snake's movement are demonstrated on a grid graph.}
    \label{fig: game: moves on grid}
\end{figure}

 An exception to the rules is made when the snake has length~$\ell\leq 2$. For these shorter lengths, we do not allow the head to move to the tail vertex, since this would allow the snake to turn around, occupying the same edge twice.

 \subsection{Eating an apple}
 At the start of the game, the location of the first apple can be any vertex. Let~$a$ be the first apple location, then the snake automatically starts on~$a$ with~$S^0=(a)$. The next apple is placed on some vertex~$a'\neq a$. The game continues as follows. At all times, there is exactly one apple on the graph. The snake eats the apple by moving its head to this vertex, at which point a new apple is immediately placed on one of the unoccupied vertices. We sometimes refer to the apple location as the apple itself.
 
 Let~$S^t=(s_1^t,\ldots, s_{\ell}^t)$ be the current snake position and $a$ be the current apple location. Suppose the head moves from~$s_1^t$ to~$a$. Then the snake eats the apple on~$a$ and the next snake position will be~$S^{t+1}= (a,s_1^t,\ldots, s_{\ell}^t)$. So different to a ``normal'' movement, $s_{\ell}^t$ is not removed from the path and the snake grows one vertex longer.

 \begin{obsv}\label{obs: apple on snake}
Let~$a$ be the location of the apple at time~$t$, then~$a\notin S^t$.
\end{obsv}

\subsection{Winning and losing conditions}\label{sec: win conditions}
 If the snake manages to reach length~$|V|$, then we say that the snake wins. Note that when this happens, the snake will occupy all the vertices. Hence, there are no more vertices an apple could be placed on.

 If there is no vertex the snake can move to, then the snake loses. To be more precise, if~$\ell<|V|$, the snake is in the position~$S=(s_1,\ldots, s_{\ell})$, and 
 \begin{align*}
     N(s_1)\cap\left(\overline{S}\cup{s_{\ell}}\right)=\emptyset,
 \end{align*}
 then the snake loses. 
 
 We want to avoid strategies where the snake can keep moving around without making progress. Imagine, for example, a scenario where the snake will lose if it eats the apple, but it can keep moving in cycles, postponing its loss forever. To this end, the snake also loses if it repeats a previous position. More precisely, if, for the current position~$S^t$, there is some~$t'<t$ such that~$S^{t'}=S^t$, then the snake loses. Note that the snake can never repeat a position it was in at a shorter length. Furthermore, any strategy that repeats a position can be reformulated as one where the position is not repeated: we simply remove the set of moves between the two identical positions. Because of this, we will not be very careful with this rule when formulating winning snake strategies. If we find a winning strategy that violates it, we know there exists a winning strategy that does adhere to this rule.

 Our aim is to determine on which graphs the snake can always win, regardless of the apple placement. To this end, we view it as a two-player game: one player controls the snake, and the other places the apples. While the \textit{snake} tries to grow to length~$|V|$, the \textit{apple placer} tries to prevent the snake from doing so. We usually assume that both players play perfectly and always use a winning strategy if possible. When the snake has a winning strategy on a graph, we call the graph \textit{snake-winnable}. We refer to the problem of determining whether a graph is snake-winnable as the \textit{snake problem}.

To conclude our overview of the game, we make a few observations regarding snake-winnable graphs. A \textit{Hamiltonian path} is a simple path that visits every vertex in $V$. Similarly, a \textit{Hamiltonian cycle} is a simple cycle that visits every vertex, and we call a graph \textit{Hamiltonian} if it has a Hamiltonian cycle. Recall that we assume that $|V|\geq 3$.

\begin{obsv}\label{obs: no Ham path}
If~$G$ does not contain a Hamiltonian path, then~$G$ is not snake-winnable.
\end{obsv}

\begin{obsv}\label{obs: ham-win}
If~$G$ is Hamiltonian, then~$G$ is snake-winnable.
\end{obsv}

\begin{obsv}\label{obs: bipartite ham path}
    Let~$G=(X\cup Y, E)$ be a bipartite graph. If~$\left||X|-|Y|\right|>1$, then~$G$ is not snake-winnable.
\end{obsv}

\begin{obsv}\label{obs: degree_one}
For a graph~$G=(V,E)$, if there is some~$v\in V$ with degree~$d(v)=1$, then~$G$ is not snake-winnable.
\end{obsv}

\section{The complexity of Snake on a graph}\label{sec: complexity}
In this section, we show that the snake problem is $\mathsf{NP}$-hard, even when restricted to grid graphs. To do this, we first characterize the odd-sized bipartite graphs that are snake-winnable. We then use this characterization to formulate a reduction from the Hamiltonian cycle problem on grid graphs to the snake problem on odd-sized grid graphs.

\begin{definition}
    A \textbf{rectangular grid graph}~$G=(V,E)$ is a graph that has an embedding with~$V = [m] \times [n]$ with~$m,n\in \mathbb{N}$. For any~$u,v\in V$ we have ~$uv\in E$ if and only if~$\|u-v\|=1$.
\end{definition}

\begin{definition}
    A graph~$G=(V,E)$ is a \textbf{grid graph} if and only if it is a vertex-induced subgraph of some rectangular grid graph.
\end{definition}

\begin{definition}
    The \textbf{theta graph}~$\Theta(p,q,r)$ is constructed by taking two vertices~$u$ and~$v$ and connecting them by three internally disjoint paths of lengths~$p$,~$q$ and~$r$. If at least one of~$p$,~$q$ or~$r$ is 0, then~$u=v$.
\end{definition}

\begin{figure}[t]
    \centering
    \includegraphics[width=0.5\textwidth]{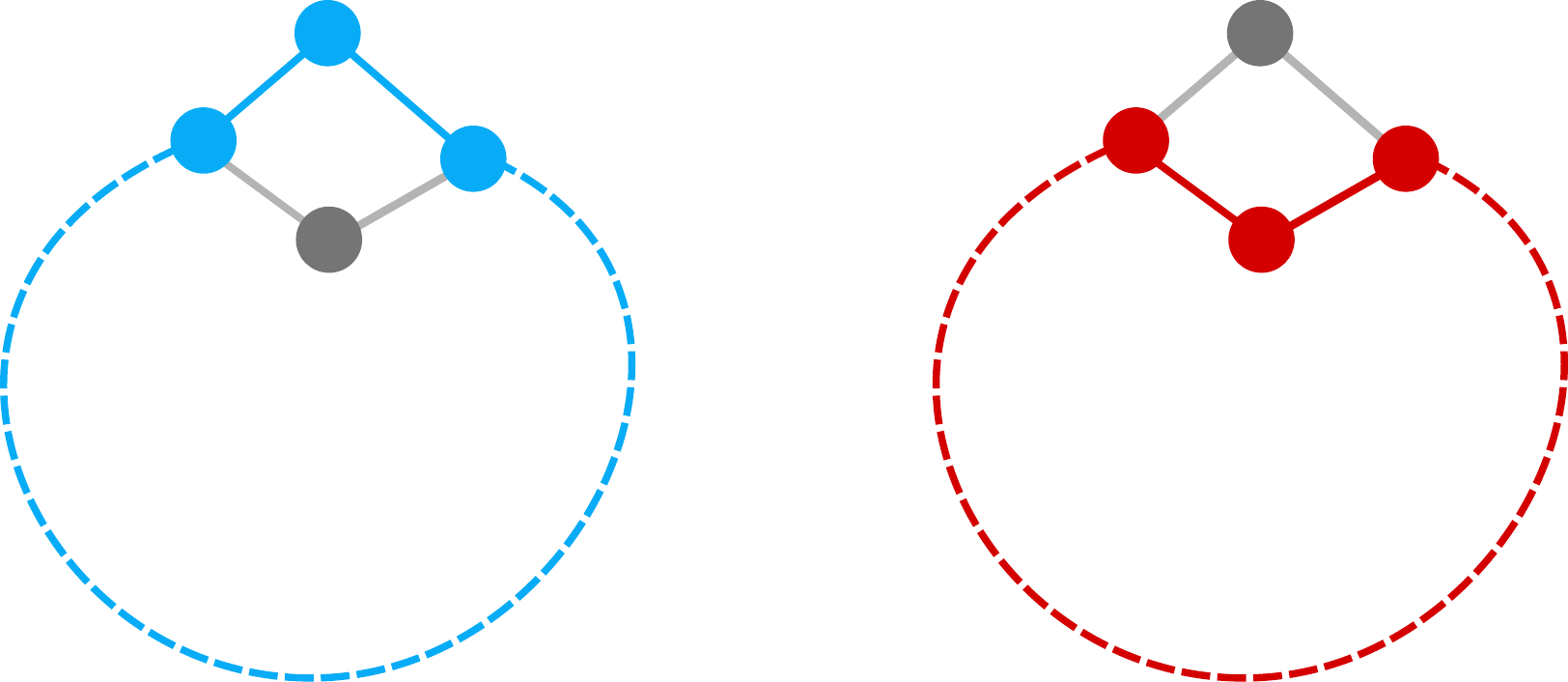}
    \caption{The snake always moves along one of the two cycles of length $|V|-1$ in the $\Theta(|V|-3,2,2)$ subgraph.}
    \label{fig: complex: theta two cycles}
\end{figure}

\begin{lemma}\label{lemma: theta win}
    If a graph~$G=(V,E)$ has~$\Theta(|V|-3,2,2)$ as a spanning subgraph, then it is snake-winnable.
\end{lemma}
\begin{proof}
The snake can employ the following strategy, which is depicted in Figure~\ref{fig: complex: theta two cycles}. Until it has length~$|V|-1$, the snake always moves along one of the two cycles that consist of the path of length~$|V|-3$, and one of the paths of length~$2$. This means that there is always only one vertex that is not on this cycle, namely the vertex on the remaining path of length 2. If the apple is placed on this vertex, then the snake switches to the other cycle that does contain this vertex. \qed
\end{proof}

\begin{figure}[t]
    \centering
    \includegraphics[width=0.4\textwidth]{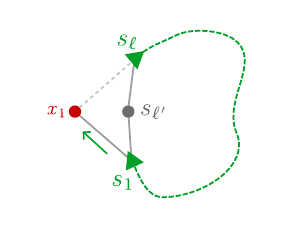}
    \caption{The snake is on a cycle of length $\ell +1$ right before it eats the second to last apple.}
    \label{fig: complex: odd bipartite theta}
\end{figure}

\begin{theorem}\label{thm: odd bipartite}
    Let~$G=(V, E)$ be a bipartite graph with partition~$V=X\cup Y$ and $|V|$ odd. Then~$G$ is snake-winnable if and only if~$\Theta(|V|-3,2,2)$ is a spanning subgraph of~$G$.
\end{theorem}
\begin{proof}[Proof sketch.]
    If~$\Theta(|V|-3,2,2)$ is a spanning subgraph of~$G$, then the snake can use the strategy from Lemma~\ref{lemma: theta win} to win.
    
     It remains to show that if there is no~$\Theta(|V|-3,2,2)$ spanning subgraph, then the apple placer has a winning strategy. We give a brief overview of this strategy, the full proof can be found in Appendix \ref{appendix: complexity}. 
     
     By cleverly placing the apples, the apple placer can guarantee the snake is on a cycle of length $\ell +1$ right before it eats the second to last apple. This is depicted in Figure \ref{fig: complex: odd bipartite theta}. After eating this apple, both the head and the remaining unoccupied vertex are adjacent to the previous head position. Hence, they must be in the same part and cannot be adjacent. It follows that the head can only move to the tail. However, by doing so, the snake would create a~$\Theta(|V|-3,2,2)$ spanning subgraph. Thus, there is no vertex the head can move to after eating the second to last apple, and thus the snake will lose.\qed
\end{proof}

\begin{corollary}\label{cor: odd grid graph}
    Let~$G=(V, E)$ be a grid graph with $|V|$ odd. Then~$G$ is snake-winnable if and only if~$\Theta(|V|-3,2,2)$ is a spanning subgraph of~$G$.
\end{corollary}

We will use the characterization from Corollary \ref{cor: odd grid graph} to show that the snake problem on grid graphs is $\mathsf{NP}$-hard. To do so, we reduce from the Hamiltonian cycle problem on grid graphs, which is $\mathsf{NP}$-complete \cite{Itai1982}.

\begin{theorem}\label{thm: NP-hard grid}
    The snake problem is $\mathsf{NP}$-hard, even restricted to grid graphs.
\end{theorem}
\begin{proof}[Proof sketch.]
    We give a general description of the reduction, further details can be found in Appendix \ref{appendix: complexity}. Given some instance~$G$ of the Hamiltonian cycle problem on grid graphs, we create a new graph~$G'$ by attaching the gadget depicted in Figure \ref{fig: complex: reduction grid graph}. By construction, $G'$ has a~$\Theta(|V|-3,2,2)$ spanning subgraph if and only if $G$ has a Hamiltonian cycle. \qed
\end{proof}

\begin{figure}[t]
     \centering
     \subfigure[The original graph~$G$.]{\includegraphics[width = 0.35\textwidth]{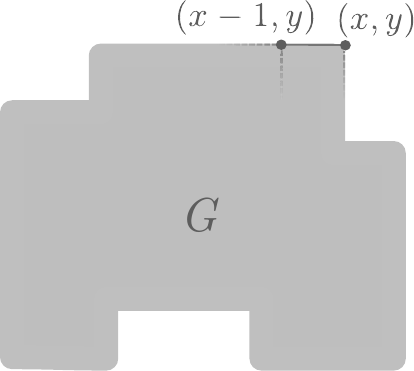}\label{fig: complex: top two grid}}\hspace{15mm}
     \subfigure[Attaching the gadget to~$G$ to create~$G'$.]{\includegraphics[width = 0.35\textwidth]{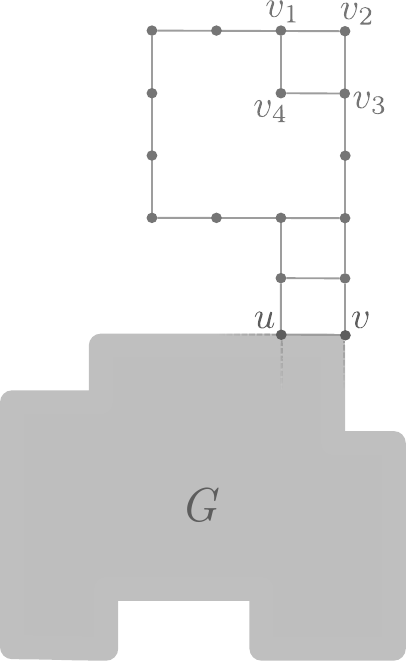}\label{fig: complex: attach gadget}}
    \caption{The reduction from a grid graph~$G$ to~$G'$ by using a gadget.}
    \label{fig: complex: reduction grid graph}
\end{figure}

We leave open whether the snake problem is in $\mathsf{NP}$,
as it seems unclear to us whether there is a compact and efficiently verifiable description of winning strategies.

\section{The girth of snake-winnable graphs}\label{sec: girth}
In this section, we prove that any non-Hamiltonian snake-winnable graph has a girth of at most 6. We also show this bound is tight by giving an example of a non-Hamiltonian snake-winnable graph that has a girth of 6.

\begin{definition}
    The \textbf{girth} of a graph~$G$, denoted by~$g(G)$, is the length of the shortest cycle in~$G$.
\end{definition}

Let~$G=(V,E)$ be a graph and~$C$ be some cycle in~$G$. We say a cycle~$C$ \textit{contains} the snake if the entire path formed by the snake lies on~$C$. To prove that any non-Hamiltonian graph with a girth greater than 6 is not snake-winnable, we use the following. When the snake reaches length~$|V|-3$, there are four possible scenarios: the snake can be contained in a cycle of length~$|V|-3$, a cycle of length~$|V|-2$, a cycle of length~$|V|-1$, or in no cycle. We will show that in each of these scenarios, the apple placer has a winning strategy.
Note that there can be no cycle of length~$|V|$ that contains the snake since this would make the graph Hamiltonian.

The following lemma shows that, as the snake grows longer, the girth required to make certain types of moves becomes smaller.

\begin{figure}[t]
\includegraphics[width = 8.5cm]{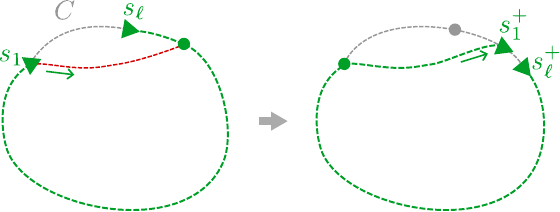}
\centering
\caption{The snake leaves~$C$ from~$s_1$ and returns to~$C$ at~$s_1^+$.}
\label{fig: girth: path outside C one}
\end{figure}

\begin{lemma}\label{lemma: leave cycle}
Let~$C$ be a cycle in~$G$ that contains the snake and let~$\ell$ be the current length of the snake. Suppose the head of the snake leaves~$C$ and returns to~$C$ after visiting~$m$ vertices in~$\overline{C}$. Then~$g(G)\leq|C|-\ell+2m+2$.
\end{lemma}
\begin{proof}
    The idea behind the proof is as follows. The path the snake takes outside of~$C$, together with the section of~$C$ it ``skips'', forms a cycle. But the snake cannot just return anywhere on~$C$, since it must move to an unoccupied vertex or the tail. Hence, the length of the section it can skip is bounded. We will show that the length of this skipped section is at most~$|C|-\ell+m+1$. Combined with the path the snake takes through~$\overline{C}$, this will result in a cycle of length at most~$|C|-\ell+2m+2$.
    
    We first note that if~$\ell \leq 2m+2$, then~$|C|-\ell+2m+2\geq |C|$. Since the existence of~$C$ implies~$g(G)\leq |C|$, the statement is trivially true. Hence, we may assume that~$\ell > 2m+2> m$, which means that while the head moves through~$\overline{C}$, the tail remains on~$C$.
    
    Suppose the head of the snake leaves~$C$ and returns to~$C$ after visiting~$m$ vertices in~$\overline{C}$, as depicted in Figure~\ref{fig: girth: path outside C one}. Let~$S=(s_1,\ldots, s_\ell)$ be the position of the snake right before the head leaves~$C$. Similarly, let~$S^-=(s_1^-,\ldots, s_\ell^-)$ be the position of the snake right before the head returns to~$C$ and~$S^+=(s_1^+,\ldots, s_\ell^+)$ the position right after. 
    By~$P_C$ we denote the section of~$C$ that the snake skips, as depicted in Figure~\ref{fig: girth: path outside C two}. More precisely,~$P_C$ is the path the head would have taken had it stayed on~$C$, with endpoints~$s_1$ and~$s_1^+$. When the snake is in position~$S^-=(s_1^-,\ldots, s_\ell^-)$, all the unoccupied vertices on~$C$ are between~$s_1$ (from where the snake left~$C$) and~$s_\ell^-$. From~$s_1^-$, the head will move to one of these unoccupied vertices or to~$s_\ell^-$. Hence, the length of~$P_C$ is maximized if the head returns to~$C$ by moving to its tail and~$s_1^+=s_\ell^-$. This is depicted in Figure~\ref{fig: girth: path outside C three}.
    
    Since the snake makes~$m$ moves before returning to~$C$, we know that~$s_\ell^- = s_{\ell -m}$. Thus, the length of~$P_C$ is maximized if~$s_1^+=s_{\ell -m}$. In this case,~$P_C$ consists of~$(s_\ell,\ldots, s_{\ell -m})$ and the path from~$s_1$ to~$s_\ell$ through the section of~$C$ that was unoccupied by~$S$. The section~$(s_\ell,\ldots, s_{\ell -m})$ has length~$m$. Before the head leaves~$C$, there are~$|C|-\ell$ unoccupied vertices between~$s_1$ and~$s_\ell$. This forms a path of length~$|C|-\ell+1$. Combined, we obtain that~$P_C$ has length at most~$|C|-\ell+1+m$.
    
   The head moves from~$s_1$ to~$s_1^+$ by visiting~$m$ vertices in~$\overline{C}$. It follows that there must also exist some~$s_1 s_1^+$-path of length at most~$m+1$ with all internal vertices in~$\overline{C}$, as depicted in Figure~\ref{fig: girth: path outside C three}. We will call this path~$P_{\overbar{C}}$, as depicted in Figure~\ref{fig: girth: path outside C two}. By combining~$P_C$ and~$P_{\overbar{C}}$, we obtain a cycle of length at most~$|C|-\ell+2m+2$. It follows that~$g(G)\leq|C|-\ell+2m+2$. \qed
\end{proof}

\begin{figure}[t]
     \centering
     \subfigure[The paths~$P_C$ on~$C$ and~$P_{\overbar{C}}$ through~$\overline{C}$.]{\includegraphics[width = 0.31\textwidth]{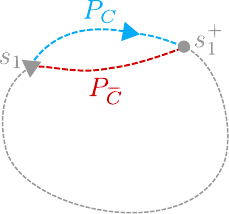}\label{fig: girth: path outside C two}}\hspace{20mm}
     \subfigure[The path $P_C$ has length at most~$|C|-\ell+1+m$ and~$P_{\overbar{C}}$ has length~$m+1$.]{\includegraphics[width = 0.32\textwidth]{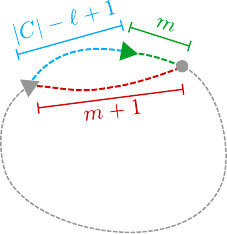}\label{fig: girth: path outside C three}}
    \caption{If the snake returns to~$C$ after visiting~$m$ vertices,~$G$ must contain a cycle of length~$|C|-\ell+2m+2$.}
    \label{fig: girth: path outside C}
\end{figure}

\begin{corollary}\label{cor: girth visit all}
Let~$G=(V,E)$ be a graph with~$g(G)> 2k$. Let~$\ell=|V|-k$ be the length of the snake, with~$k\geq 2$. Suppose~$G$ has a cycle~$C$ that contains the snake. If the head of the snake leaves~$C$, then it must visit all vertices in~$\overline{C}$ before returning to~$C$.
\end{corollary}

\begin{obsv} \label{obs: visit to win}
    Let~$\overline{S}$ be the current set of unoccupied vertices. To win, the snake will have to visit all vertices in~$\overline{S}$ at some point in the future.
\end{obsv}

Using Corollary~\ref{cor: girth visit all} and Observation \ref{obs: visit to win}, we can prove the following lemma in a similar way as
Lemma~\ref{lemma: leave cycle}.

\begin{lemma}\label{lemma: longer cycle lose}
    Let~$G=(V, E)$ be a graph with~$g(G)> 2k$. Let~$\ell=|V|-k$ be the length of the snake with~$k\geq 2$. Suppose~$G$ has a cycle~$C$ that contains the snake with~$|C| > \ell$. If the apple is on some vertex in~$\overline{C}$, then the snake will lose.
\end{lemma}

\begin{corollary}\label{cor: longer cycle lose three}
    Let~$G=(V, E)$ be a graph with~$g(G)>6$. Consider the moment the snake grows to length~$|V|-3$. If there is a cycle~$C$ of length~$|V|-1$ or~$|V|-2$ that contains the snake, then the snake will lose.
\end{corollary}

Corollary~\ref{cor: longer cycle lose three} excludes two out of the four possible scenarios for when the snake reaches length~$|V|-3$. The two that remain are: the snake is contained in a cycle of length~$|V|-3$, and there is no cycle that contains the snake. We only state the lemmas for these two cases. Full proofs can be found in Appendix \ref{appendix: girth}.

\begin{lemma}\label{lemma: three left no cycle}
    Let~$G=(V,E)$ be a non-Hamiltonian graph with~$g(G)>6$. Consider the moment the snake grows to length~$\ell = |V|-3$ and suppose there is no cycle in~$G$ that contains the snake. Then, the snake will lose.
\end{lemma}

\begin{lemma}\label{lemma: girth head to tail}
    Let~$G=(V, E)$ be a non-Hamiltonian graph with~$g(G)>6$. Consider the moment the snake grows to length~$\ell = |V|-3$ and suppose there is a cycle~$C$ of length~$|V|-3$ that contains the snake. Then, the snake will lose.
\end{lemma}                                                                                                                   
\begin{theorem} \label{thm: girth bound}
Let~$G=(V, E)$ be a non-Hamiltonian graph with~$g(G)>6$. Then,~$G$ is not snake-winnable.
\end{theorem}
\begin{proof}
    Consider the moment the snake grows to length~$|V|-3$. Since~$G$ is non-Hamiltonian, there cannot be any cycle of length~$|V|$. Hence, the snake is either contained in a cycle of length~$|V|-1$, a cycle of length~$|V|-2$, a cycle of length~$|V|-3$, or in no cycle.

    If the snake is contained in a cycle of length~$|V|-1$ or~$|V|-2$, then the snake will lose by Corollary~\ref{cor: longer cycle lose three}. If the snake is contained in a cycle of length~$|V|-3$, then it will lose by Lemma~\ref{lemma: girth head to tail}. Finally, if the snake is not contained in any cycle, then it will lose by Lemma~\ref{lemma: three left no cycle}. It follows that once the snake reaches length~$|V|-3$, the apple placer can guarantee that the snake loses, and thus~$G$ is not snake-winnable. \qed
\end{proof}

The graph in Figure \ref{fig: girth: partial grid win} shows the bound from Theorem \ref{thm: girth bound} is tight, as it is non-Hamiltonian, has girth 6, and is snake-winnable. The winning strategy for the snake can be found in Appendix \ref{appendix: girth}.

\begin{figure}[t]
\includegraphics[width = 6.6cm]{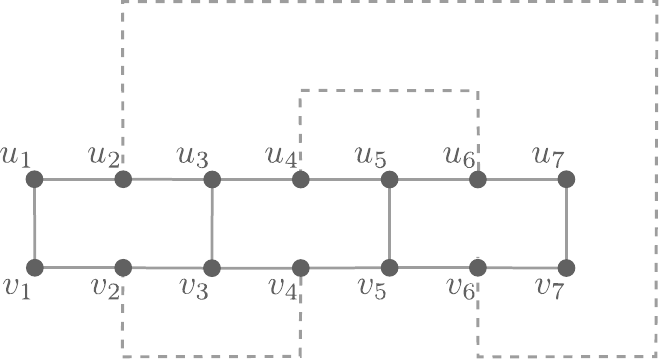}
\centering
\caption{A non-Hamiltonian graph with a girth of 6 that is snake-winnable.}
\label{fig: girth: partial grid win}
\end{figure}

\section{Snake-winnable graphs with vertex-connectivity 1}\label{sec: connectivity}
One of the challenges for the snake is that its body can block access to certain parts of the graph. Intuitively, this becomes more of an issue when the graph has low vertex-connectivity. In this section, we prove that for vertex-connectivity 1, there is only one type of graph that is snake-winnable.

\begin{definition}
    Let~$G=(V, E)$ be a connected simple graph. If there is a single vertex~$v$ such that~$G-v$ is disconnected, then~$G$ has has \textbf{vertex-connectivity}~1, which we denote as~$\kappa(G)=1$. We call~$v$ a \textbf{cut vertex} of~$G$.
\end{definition}

\begin{obsv}\label{obs:too_many_comp}
    Let~$v$ be a cut vertex of~$G$. If~$G-v$ has at least 3 components, then~$G$ is not snake-winnable.
\end{obsv}

\begin{figure}[t]
\includegraphics[width = 0.35\textwidth]{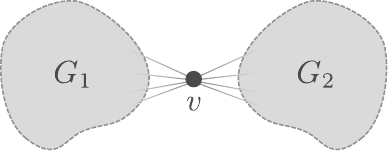}
\centering
\caption{A snake-winnable graph with cut vertex~$v$. The subgraphs~$G_1+v$ and~$G_2+v$ are both complete.}
\label{fig: connectivity: 1c}
\end{figure} 

\begin{lemma}\label{lemma: 1c win}
Let~$G=(V, E)$ be a graph with~$\kappa(G)=1$ and let~$v$ be a cut vertex of~$G$. Suppose that~$G_1$ and~$G_2$ are the only two connected components of~$G-v$ with~$|V(G_1)|=|V(G_2)|=m$ and~$m\geq2$. Furthermore, let~$G_1+v$ and~$G_2+v$ both be complete. Then~$G$ is snake-winnable.
\end{lemma}
\begin{proof}[sketch]
The graph~$G$ is depicted in Figure~\ref{fig: connectivity: 1c}. We give a sketch of the snake's strategy, full details can be found in Appendix \ref{appendix: connectivity}. The snake can ensure that once it grows to length $m+1$, it occupies exactly all of the vertices in either $G_1+v$ or $G_2+v$, with its tail on $v$. Since both $G_1+v$ and $G_2+v$ are complete, it can move to the tail vertex, placing its head on $v$. The remaining unoccupied vertices are either $V(G_1)$ or $V(G_2)$, and thus, all the remaining unoccupied vertices are adjacent to $v$. Regardless of where the next apple is placed, the snake can directly move to it from $v$ and the tail remains in place. The snake repeats this until all of the remaining vertices are occupied.
\qed
\end{proof}

It turns out the graph from Lemma~\ref{lemma: 1c win} is the only type of snake-winnable graph with vertex-connectivity 1. To show this, we use the following two lemmas that show the limitations of the snake on graphs with vertex-connectivity 1.

\begin{lemma}\label{lemma: too long 1c}
    Let~$G=(V, E)$ be a graph with~$\kappa(G)=1$ and let~$v$ be a cut vertex of~$G$. Let~$G_1$ and~$G_2$ be two different connected components of~$G-v$, with~$|V(G_1)|=m$. Suppose the head of the snake is in~$G_1$ and the snake has a length of at least~$m+2$. If there is an unoccupied vertex in~$G_2$, then the snake will lose.
\end{lemma}

\begin{lemma}\label{lemma: too long apple 1c}
    Let~$G=(V, E)$ be a graph with~$\kappa(G)=1$ and let~$v$ be a cut vertex of~$G$. Let~$G_1$ and~$G_2$ be two different connected components of~$G-v$, with~$|V(G_1)|=m$. Suppose the head of the snake is in~$G_1$, the snake has a length of at least~$m+1$, and the apple is on some vertex~$a$ in~$G_1$. If there is some unoccupied vertex in~$G_2$, then the snake will lose.
\end{lemma}

We now show that if the two components of~$G-v$ have different sizes, then the apple placer has a winning strategy.

\begin{lemma}\label{lemma: 1c diff size}
Let~$G=(V,E)$ be a graph with~$\kappa(G)=1$ and let~$v$ be a cut vertex of~$G$. Let~$G_1$ and~$G_2$ be two different components of~$G-v$. If~$|V(G_1)|\neq|V(G_2)|$ then~$G$ is not snake-winnable.
\end{lemma}

\begin{proof}
    First note that if either~$|V(G_1)|=1$ or~$|V(G_2)|=1$, then~$G$ has a vertex of degree 1 and is not snake-winnable by Observation~\ref{obs: degree_one}. Hence, we may assume that both~$G_1$ and~$G_2$ have at least two vertices.
    
    Let~$|V(G_1)|=m_1$ and~$|V(G_2)|=m_2$ and suppose~$m_1\neq m_2$. Without loss of generality, we assume that~$m_1<m_2$. 
    When the snake grows to length~$m_1-1$, there must be some vertex~$u_1$ in~$G_1$ that is unoccupied. The apple placer chooses~$u_1$ as the next apple location. When the snake eats the apple on~$u_1$, it grows to length~$m_1$. Since~$m_1<m_2$, there must be some unoccupied vertex~$u_2$ in~$G_2$. The apple placer chooses~$u_2$ as the next apple location. 
    
    When the snake eats the apple on~$u_2$, it is either entirely in~$G_2$, or both the cut vertex~$v$ and~$u_2$ are occupied by the snake. Thus, at least two of the occupied vertices are in~$G_2+v$. Since the snake now has length~$m_1+1$, it follows that there must be some unoccupied vertex~$v_1$ in~$G_1$. The apple placer chooses~$v_1$ as the next apple location. When the snake eats the apple on~$v_1$, it grows to length~$m_1+2$ with its head in~$G_1$. Furthermore, there must be an unoccupied vertex in~$G_2$ and by Lemma~\ref{lemma: too long 1c}, the snake will lose. \qed
\end{proof}

It remains to show that if either $G_1+v$ or $G_2+v$ is incomplete, then $G$ is not snake-winnable. We only give the result, the full proof can be found in Appendix~\ref{appendix: connectivity}.

\begin{lemma}\label{lemma: 1c not comp}
Let~$G=(V, E)$ be a graph with~$\kappa(G)=1$ and let~$v$ be a cut vertex of~$G$. Let~$G_1$ and~$G_2$ be two different connected components of~$G-v$ and suppose~$G_1+v$ and~$G_2+v$ are not both complete. Then~$G$ is not snake-winnable.
\end{lemma}

\begin{theorem}\label{thm: 1c}
Let~$G=(V, E)$ have vertex-connectivity 1 and let~$v$ be a cut vertex of~$G$. Furthermore, let~$G_1$ and~$G_2$ be two different connected components of~$G-v$. Then~$G$ is snake-winnable if and only if~$G-v$ has two components,~$|V(G_1)|=|V(G_2)|\geq 2$, and~$G_1+v$ and~$G_2+v$ are both complete.
\end{theorem}
\begin{proof}
By Lemma~\ref{lemma: 1c win}, if ~$G-v$ has two components,~$|V(G_1)|=|V(G_2)|\geq 2$, and~$G_1+v$ and~$G_2+v$ are both complete, then~$G$ is snake-winnable.

 By Observation~\ref{obs:too_many_comp},~$G$ is not snake-winnable if~$G-v$ has more than two components. If either~$|V(G_1)|=1$ or~$|V(G_2)|=1$, then~$G$ had a vertex of degree 1 and is not snake-winnable by Observation~\ref{obs: degree_one}. By Lemma~\ref{lemma: 1c diff size}, if~$|V(G_1)|\neq|V(G_2)|$, then~$G$ is not snake-winnable. By Lemma~\ref{lemma: 1c not comp},~$G$ is also not snake-winnable if either~$G_1+v$ or~$G_2+v$ is incomplete. \qed
\end{proof}

\section{Open problems}
Having characterized snake-winnable graphs for odd-sized bipartite graphs, a natural question is whether a similar characterization exists for bipartite graphs with an even number of vertices, or even just for grid graphs. Similarly, given the characterization of snake-winnable graphs with vertex-connectivity 1, we wonder whether the snake-winnable graphs with vertex-connectivity 2 can also be characterized. Finally, the question of whether the snake problem is in $\mathsf{NP}$ also remains open.

\bibliographystyle{splncs04}
\bibliography{refs}

\appendix
\renewcommand{\thesection}{\arabic{section}}
\setcounter{section}{1}

\section{Appendix: The game of Snake on a graph}\label{appendix: game_desc}
\setcounter{obsv}{0}
\begin{obsv}
Let~$a$ be the location of the apple at time~$t$, then~$a\notin S^t$.
\end{obsv}
\begin{proof}
First, we note that~$s_1^t \neq a$, since otherwise the apple would be eaten and no longer be on the graph.

Suppose~$a=s_i^t$ for some~$i\in\{2,\ldots,\ell\}$. Then the head must have been on~$a$ at some earlier point after which~$a$ has remained occupied by some part of the snake up until time~$t$. But we know that the apple on~$a$ could only have been placed at a time when~$a$ was unoccupied. Hence, the head must have reached~$a$ sometime after the apple was placed, which would imply that the apple was already eaten and is no longer on the graph. \qed
\end{proof}

\begin{obsv}
If~$G$ does not contain a Hamiltonian path, then~$G$ is not snake-winnable.
\end{obsv}
\begin{proof}
Since the snake must always form a simple path in~$G$, if it reaches length~$|V|$ it will form a Hamiltonian path.\qed
\end{proof}

\begin{obsv}
If~$G$ is Hamiltonian, then~$G$ is snake-winnable.
\end{obsv}
\begin{proof}
Since~$G$ is Hamiltonian, there is some simple cycle~$C$ in~$G$ that contains all vertices. The snake can keep moving along this cycle. Since any apple will be placed on~$C$, it will eventually be eaten by the snake. By following this strategy, the snake can keep growing until it covers the entire cycle, which contains all vertices.\qed
\end{proof}

\begin{obsv}
    Let~$G=(X\cup Y, E)$ be a bipartite graph. If~$\left||X|-|Y|\right|>1$, then~$G$ is not snake-winnable.
\end{obsv}
\begin{proof}
    Any path in~$G$ must alternate between~$X$ and~$Y$. Hence,~$G$ can only contain a Hamiltonian path if~$|X|=|Y|$ or~$\left||X|-|Y|\right|=1$. By Observation~\ref{obs: no Ham path}, it follows that if~$\left||X|-|Y|\right|>1$, then~$G$ is not snake-winnable.\qed
\end{proof}

\begin{obsv}
For a graph~$G=(V,E)$, if there is some~$v\in V$ with degree~$d(v)=1$, then~$G$ is not snake-winnable.
\end{obsv}
\begin{proof}
Suppose there is some vertex~$v\in V$ with~$d(v)= 1$ and let~$u$ be the only neighbor of~$v$. Since~$|V| \geq 3$, the first apple can be placed on some vertex that is not~$v$. The second apple is then placed on~$v$. When the snake eats the second apple, we must have~$s_1=v$ and~$s_2=u$. But since the snake has length 2, the head is not allowed to move to~$u$, which means there is no vertex the head can move to.\qed
\end{proof}

\section{Appendix: The complexity of Snake on a graph}\label{appendix: complexity}
\setcounter{theorem}{4}
\begin{theorem}\label{appendix: thm: odd bipartite}
    Let~$G=(V, E)$ be a bipartite graph with partition~$V=X\cup Y$ and $|V|$ odd. Then~$G$ is snake-winnable if and only if~$\Theta(|V|-3,2,2)$ is a spanning subgraph of~$G$.
\end{theorem}
\begin{proof}
    If~$\Theta(|V|-3,2,2)$ is a spanning subgraph of~$G$, then the snake can use the strategy from Lemma~\ref{lemma: theta win} to win.
    
     It remains to show that if there is no~$\Theta(|V|-3,2,2)$ spanning subgraph, then the apple placer has a winning strategy. We will approach this as follows. As the snake reaches length~$|V|-1$, there are at most two vertices the snake can move to: the only remaining unoccupied vertex and the tail vertex. If the head and the unoccupied vertex are in the same part, then the head can only move to the tail. By cleverly placing the apples, we will show that the apple placer can guarantee this is the case. Furthermore, due to the previous apple placement, the snake will create a~$\Theta(|V|-3,2,2)$ spanning subgraph if it moves to the tail. Hence, the snake cannot move to the unoccupied vertex nor the tail vertex, and will thus lose.
    
    Assume~$\Theta(|V|-3,2,2)$ is not a spanning subgraph of~$G$.
    Since~$G$ is odd-sized, by Observation~\ref{obs: bipartite ham path} we know~$G$ can only be snake-winnable if~$\left||X|-|Y|\right|=1$. Without loss of generality, we will assume that~$|X|=|Y|+1$. 
    
\begin{figure}[t]
\includegraphics[width = 5.5cm]{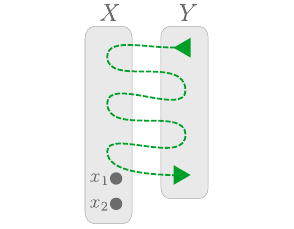}
\centering
\caption{The snake always alternates between~$X$ and~$Y$. When the snake reaches length~$|V|-2$ with its head in~$Y$, the two remaining unoccupied vertices are in~$X$.}
\label{appendix: fig: grid: odd bipartite two left}
\end{figure}

    Consider the moment the snake reaches length~$|V|-3$. Since~$|V|-3$ is even, the snake occupies the same number of vertices of~$X$ as of~$Y$. Hence, there is one vertex~$y\in Y$ that is unoccupied. The apple placer places the next apple on~$y$. When the snake eats the apple on~$y$, it reaches length~$|V|-2$, which is an odd number. Since the snake alternates between~$X$ and~$Y$,  both its head and tail must be in~$Y$ as depicted in Figure~\ref{appendix: fig: grid: odd bipartite two left}. Hence, the two remaining unoccupied vertices are both in~$X$.

    Suppose both unoccupied vertices are adjacent to the head and the tail. Then between the head and tail of the snake, we have the path of length~$|V|-3$ formed by the snake itself, and two paths of length 2, each through one of the unoccupied vertices. Thus, we obtain a~$\Theta(|V|-3,2,2)$ spanning subgraph of~$G$. Since we assumed such a spanning subgraph did not exist, we can conclude that at least one of the unoccupied vertices is not adjacent to both the head and the tail. Let this be~$x_1$ and let~$x_2$ be the other unoccupied vertex. The apple placer places the next apple on~$x_1$.

    Suppose the snake immediately eats~$x_1$, without moving to some other vertex first. Then the tail would remain in the same place. Furthermore, this requires~$x_1$ to be adjacent to the head, which means it is not adjacent to the tail. Hence, from~$x_1$ the snake can only move to another unoccupied vertex. But since~$x_2$ is the only remaining unoccupied vertex and~$x_1$ and~$x_2$ are both in~$X$, this is not possible. It follows that the snake must move to some other vertex first before eating the apple on~$x_1$. At length~$|V|-2$, the snake cannot move to its tail since this would create an odd cycle. Thus, the only move the snake can make is to~$x_2$.

    After moving to~$x_2$, the previous tail vertex becomes unoccupied, as depicted in Figure~\ref{appendix: fig: grid: odd bipartite move to prev}. From~$x_2$, the snake cannot move to~$x_1$ or the current tail vertex. Hence, it has to move to the previous tail vertex. Continuing this reasoning, we find that the only thing the snake can do until it eats the apple is repeatedly moving to the previous tail vertex. Of course, this type of move requires that the previous tail vertex is adjacent to the current head. Since the previous tail vertex is also adjacent to the current tail, the snake must be moving along a cycle of length~$|V|-1$, consisting of the snake itself and the previous tail vertex. 

    Let~$S$ be the position of the snake right before it moves to~$x_1$, which is depicted in Figure~\ref{appendix: fig: grid: odd bipartite theta}. Let~$s_1$ be the head vertex and~$s_{\ell}$ the tail vertex. Furthermore, let~$s_{\ell'}$ be the previous tail vertex, that is adjacent to both~$s_1$ and~$s_{\ell}$. Note that since the head moves from~$s_1$ to~$x_1$,~$s_1$ and~$s_\ell$ must both be in~$Y$, and~$s_{\ell'}$ in~$X$. After eating the apple on~$x_1$,~$s_\ell$ remains the tail vertex. From~$x_1$, the snake cannot move to~$s_{\ell'}$, since they are both in~$X$. If the snake can move to~$s_{\ell}$, then the paths~$(s_1, x_1, s_{\ell})$ and~$(s_1, s_{\ell'}, s_{\ell})$, together with~$S$ form a~$\Theta(|V|-3,2,2)$ spanning subgraph of~$G$. Since we assumed such a spanning subgraph did not exist, we can conclude that there is no vertex the snake can move to form~$x_1$, and will thus lose. \qed
\end{proof}

 \begin{figure}[t]
     \centering
     \subfigure[The snake has to move to~$x_2$ first. From~$x_2$, the snake can only move to the previous tail.]{\includegraphics[width = 0.56\textwidth]{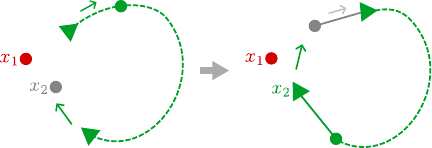}\label{appendix: fig: grid: odd bipartite move to prev}}\hspace{5mm}
     \subfigure[The position of the snake before eating the apple on~$x_1$.]{\includegraphics[width = 0.38\textwidth]{Figures/grid/odd_bipartite_theta.pdf}\label{appendix: fig: grid: odd bipartite theta}}
    \caption{After eating the apple on~$x_1$, the snake can only move to~$s_l$, but this move requires~$\Theta(|V|-3,2,2)$ spanning subgraph.}
    \label{appendix: fig: grid: odd bipartite}
\end{figure}

 \begin{figure}[t]
     \centering
     \subfigure[The two rightmost points of the top row of~$G$.]{\includegraphics[width = 0.35\textwidth]{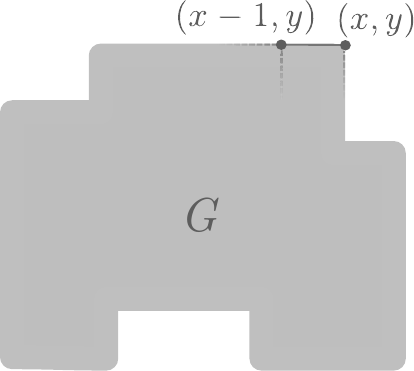}\label{appendix: fig: complex: top two grid}}\hspace{15mm}
     \subfigure[Attaching the gadget to~$G$ to create~$G'$.]{\includegraphics[width = 0.35\textwidth]{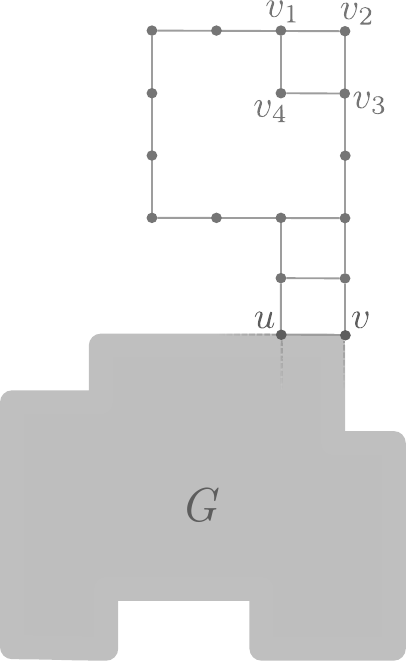}\label{appendix: fig: complex: attach gadget}}
    \caption{The reduction from a grid graph~$G$ to~$G'$ by using a gadget.}
    \label{appendix: fig: complex: reduction grid graph}
\end{figure}

\setcounter{theorem}{6}
\begin{theorem}
    The snake problem is $\mathsf{NP}$-hard, even when restricted to grid graphs.
\end{theorem}
\begin{proof}
    Let~$G=(V,E)$ be an instance of the Hamiltonian cycle problem on grid graphs. Since odd-sized grid graphs are never Hamiltonian, we may assume~$G$ is even-sized. Let~$v=(x,y)$ be the vertex on the top row of~$G$ that is furthest to the right, as depicted in Figure~\ref{appendix: fig: complex: top two grid}. This means the points~$(x,y+1)$ and~$(x+1,y)$ are not vertices of~$G$. If~$(x-1,y)$ is also not a vertex of~$G$, then~$(x,y)$ has degree at most 1, which means~$G$ cannot be Hamiltonian. Hence we may assume that~$(x-1,y)$ is also a vertex of~$G$, and we denote~$u=(x-1,y)$. Since the degree of~$v$ is at most 2, the edge~$uv$ must be part of any Hamiltonian cycle.

    For our reduction, we create a new graph~$G'=(V',E')$ by attaching a gadget to~$u$ and~$v$, as depicted in Figure~\ref{appendix: fig: complex: attach gadget}. Since the gadget is odd-sized and we assumed~$G$ was even-sized,~$G'$ is an odd-sized grid graph. In~$G'$, let the vertices~$v_1,v_2,v_3,v_4$ be as depicted in Figure~\ref{appendix: fig: complex: attach gadget}. Note that there is no cycle in~$G'$ of length~$|V'|-1$ that contains both~$v_2$ and~$v_4$. Thus, any~$\Theta(|V|-3,2,2)$ spanning subgraph of~$G'$ must consist the following three paths:~$(v_1,v_2,v_3)$,~$(v_1,v_4,v_3)$, and some~$v_1v_3$-path that contains all vertices in~$V'\backslash\{v_2,v_4\}$. Since the gadget only connects to~$G$ at~$u$ and~$v$, to form the latter path~$G$ must have a Hamiltonian path from~$u$ to~$v$. But since the edge between~$uv$ is part of any Hamiltonian cycle, this Hamiltonian path from~$u$ to~$v$ exists if and only if~$G$ is Hamiltonian. It follows that~$G'$ has a~$\Theta(|V|-3,2,2)$ spanning subgraph if and only if~$G$ is Hamiltonian. By Corollary~\ref{cor: odd grid graph}, we obtain that~$G'$ is snake-winnable if and only if~$G$ is Hamiltonian, completing our reduction and proving the $\mathsf{NP}$-hardness of the snake problem on grid graphs.\qed
\end{proof}

\section{Appendix: The girth of snake-winnable graphs}\label{appendix: girth}
\setcounter{lemma}{2}

\begin{corollary}\label{appendix: cor: girth visit all}
Let~$G=(V,E)$ be a graph with~$g(G)> 2k$. Let~$\ell=|V|-k$ be the length of the snake, with~$k\geq 2$. Suppose~$G$ has a cycle~$C$ that contains the snake. If the head of the snake leaves~$C$, then it must visit all vertices in~$\overline{C}$ before returning to~$C$.
\end{corollary}
\begin{proof}
    Let~$m$ be the number of vertices that the head visits before returning to~$C$. By Lemma~\ref{lemma: leave cycle}, we know~$g(G)\leq |C|-\ell+2m+2$. Suppose the head does not visit all the vertices in~$\overline{C}$, in other words, we have~$m \leq |V|-|C|-1$. This gives us the following upper bound for~$g(G)$:
    \begin{align*}
        g(G)\leq |C|-\ell+2m+2 \leq 2|V|-|C|-\ell.
    \end{align*}
    Since we have~$\ell = |V|-k$, it follows that
    \begin{align*}
        g(G)\leq |V|-|C| +k.
    \end{align*}
    Since~$C$ contained the snake, we know that~$|C|\geq |V|-k$, and thus we have~$|V|-|C|\leq k$. This gives us~$g(G)\leq 2k$, which contradicts that~$g(G)> 2k$. We conclude that the head must visit all vertices in~$\overline{C}$ before returning to~$C$. \qed
\end{proof}

\begin{obsv} \label{appendix: obs: visit to win}
    Let~$\overline{S}$ be the current set of unoccupied vertices. To win, the snake will have to visit all vertices in~$\overline{S}$ at some point in the future.
\end{obsv}
\begin{proof}
    To win, the snake has to obtain a position that occupies all vertices. If a vertex is currently unoccupied, then it can only become occupied if the snake visits it.\qed
\end{proof}

\begin{lemma}\label{appendix: lemma: longer cycle lose}
    Let~$G=(V, E)$ be a graph with~$g(G)> 2k$. Let~$\ell=|V|-k$ be the length of the snake with~$k\geq 2$. Suppose~$G$ has a cycle~$C$ that contains the snake with~$|C| > \ell$. If the apple is on some vertex in~$\overline{C}$, then the snake will lose.
\end{lemma}
\begin{proof}
    Suppose the apple is on some vertex in~$\overline{C}$. Then at some point, the head of the snake will have to leave~$C$ to eat the apple. We will first show that after the head leaves~$C$, it can never return to~$C$.
    
    Suppose the head leaves~$C$ and returns to~$C$ at some later point. By Corollary~\ref{cor: girth visit all}, the snake has to visit all vertices in~$\overline{C}$ before it can return. Let~$m$ be the number of vertices the snake visits before returning to~$C$, then~$m=|V|-|C|$. Furthermore, one of these vertices must contain the apple.
    
   When the snake eats the apple, the tail will not move. So while the snake visits~$m$ vertices in~$\overline{C}$, the tail will only move~$m-1$ times. It follows that right before the head moves back to~$C$, the tail will be on~$s_{\ell-m+1}$. Let~$s_1^+$ be the position of the head right after it re-enters~$C$. We now use similar reasoning as for Lemma~\ref{lemma: leave cycle} to obtain that there is a path on~$C$ from~$s_1$ to~$s_1^+$ of length at most~$|C|-\ell+m+1$.
   
    Combined with the path of the head through~$\overline{C}$, this gives us a cycle of length at most~$|C|-\ell+2m+1$. Given that~$m=|V|-|C|$, we obtain that
    \begin{align*}
        |C|-\ell+2m+1 = 2|V|-|C|-\ell+1.
    \end{align*}
    Since~$|C|\geq \ell+1$, we have
    \begin{align*}
        2|V|-|C|-\ell+1 \leq 2(|V|-\ell),
    \end{align*}
    and since~$\ell\geq |V|-k$, it follows that~$G$ must contain a cycle of length at most~$2k$. This contradicts that~$g(G)> 2k$, and thus we can conclude that the head cannot return to~$C$.

    Since~$|C|\geq \ell+1$, there must be some vertex~$v$ on~$C$ that is unoccupied right before the head leaves~$C$. Since the snake can no longer return to~$C$, the snake can never visit~$v$ again, and~$v$ will remain unoccupied. By Observation~\ref{obs: visit to win}, the snake will lose. \qed
\end{proof}

\begin{corollary}\label{appendix: cor: longer cycle lose three}
    Let~$G=(V, E)$ be a graph with~$g(G)>6$. Consider the moment the snake grows to length~$|V|-3$. If there is a cycle~$C$ of length~$|V|-1$ or~$|V|-2$ that contains the snake, then the snake will lose.
\end{corollary}
\begin{proof}
    Since the snake is contained in~$C$ and~$|C|<|V|$, there is some unoccupied vertex in~$\overline{C}$. The apple placer places the next apple on this vertex. We can now use Lemma~\ref{lemma: longer cycle lose} with~$k=3$ to obtain that the snake will lose.\qed
\end{proof}

\begin{figure}[t]
     \centering
     \subfigure[The snake can take~$P_u$ or~$P_v$ to~$S$.]{\includegraphics[width = 0.39\textwidth]{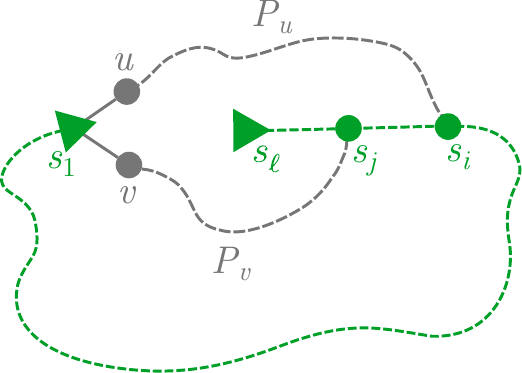}\label{appendix: fig: girth: no cycle paths}}\hspace{10mm}
     \subfigure[The snake in position~$S'$.]{\includegraphics[width = 0.39\textwidth]{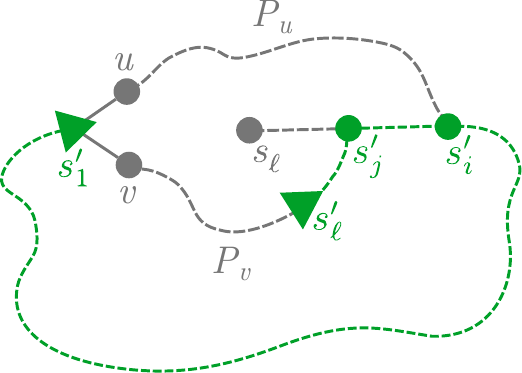}\label{appendix: fig: girth: no cycle shifted}}
    \caption{There is no cycle that contains the snake and the head can move to either~$u$ or~$v$ next, without ensuring a loss.}
    \label{appendix: fig: girth: no cycle}
\end{figure}

\setcounter{lemma}{9}
\begin{lemma}\label{appendix: lemma: girth no cycle}
   Let~$G=(V, E)$ be a graph with~$g(G)>2k$ and let~$\ell=|V|-k$ be the length of the snake with~$k\geq 2$. Suppose there is no cycle in~$G$ that contains the snake. Then, there is at most one vertex to which the snake can move next without ensuring a loss.
\end{lemma}
\begin{proof}
    The idea behind the proof is as follows. If there are two different vertices the snake can move to, then this will give us two different paths through~$\overline{S}$. We will see that these paths must both lead to~$S$ and cannot meet before they reach~$S$. With one of these paths, we then construct a cycle that contains a big portion of~$S$. Using this cycle and the remaining path, we will show that the snake can move in a way that contradicts Corollary~\ref{cor: girth visit all}.
    
    Let~$S$ be the current position of the snake with~$s_1$ the head position, and~$s_{\ell}$ the tail position. Let~$u$ and~$v$ be two unoccupied vertices that are both adjacent to the head. Suppose the snake can move to either~$u$ or~$v$ next, without ensuring a loss. First, suppose the snake moves to~$u$. Then it cannot move to~$v$ before the head moves to~$S$ again, otherwise we would have a cycle of length at most~$k+1$. By Observation~\ref{obs: visit to win} the snake will eventually have to visit~$v$ to win. It follows that the snake has to move to~$S$ at some point. With the same argument, this is also the case if the snake moves to~$v$ first. 
    
    Let~$P_u$ and~$P_v$ be the paths the head can take from~$u$ and~$v$ to~$S$ respectively.  The two paths are depicted in Figure~\ref{appendix: fig: girth: no cycle paths}. We know~$P_u$ and~$P_v$ must be internally disjoint, otherwise, we would have a cycle of length at most~$k+1$. Let~$s_i$ be the endpoint of~$P_u$ and~$s_j$ of~$P_v$. Without loss of generality, we may assume that~$i\leq j$. Since there is no cycle that contains~$S$, we must have~$i,j< \ell$. This means~$s_{\ell}$ lies on neither of the two paths.

    Suppose the snake takes the path~$P_v$. When the head moves to~$s_j$, the cycle~$C= (s_{j-1},\ldots s_1, P_u)$ will contain the snake. The snake can now keep moving along~$C$ until the head is on~$s_1$ again. Let~$S'$ be the position of the snake at this moment, which is depicted in Figure~\ref{appendix: fig: girth: no cycle shifted}. Note that~$C$ still contains~$S'$. This new position is similar to the original position~$S$, but the last section of the snake, from~$s_{j+1}'$ to~$s_l'$, now lies on~$P_v$. Since we assumed the snake could also take~$P_u$ from~$u$ to~$s_i$, it follows that from the position~$S'$, the snake can take the path~$P_u$ to leave and re-enter~$C$. However, we know~$s_{\ell}$ neither lies on this path, nor on the cycle~$C$. By Corollary~\ref{cor: girth visit all}, this is not possible. It follows that from the original position~$S$, either moving to~$u$ or moving to~$v$ must result in a loss. \qed
\end{proof}

By Lemma~\ref{appendix: lemma: girth no cycle}, we obtain that if~$G$ has a girth of at least 7, the snake has length~$|V|-3$, and there is no cycle containing the snake, then we can predict its movement. We can use this to our advantage when describing a strategy for the apple placer. The following lemma will be useful for this purpose.

\begin{lemma} \label{appendix: lemma: girth two left}
Let~$G=(V, E)$ be a non-Hamiltonian graph with~$g(G)>4$. When the snake grows to length~$|V|-2$, if the two unoccupied vertices are not adjacent, then the snake will lose.
\end{lemma}
\begin{proof}
Since~$G$ is non-Hamiltonian, there is either a cycle of length~$|V|-2$, a cycle of length~$|V|-1$, or no cycle that contains the snake. By Lemma~\ref{lemma: longer cycle lose}, the snake will lose if there is a cycle of length~$|V|-1$ that contains the snake.

First, suppose that there is a cycle of length~$|V|-2$ that contains the snake. In other words, the head of the snake is adjacent to the tail, and the snake itself forms a cycle. Let~$C$ be this cycle. Both unoccupied vertices are in~$\overline{C}$, and the apple has to be placed on one of these two vertices. The snake cannot eat the apple by only moving along~$C$. Hence, at some point, it must move to one of the unoccupied vertices. When it does so, by Corollary~\ref{cor: girth visit all}, the snake has to visit the other unoccupied vertex next. However, this is not possible, since the unoccupied vertices are not adjacent. It follows that the snake will lose.

Now suppose there is no cycle that contains the snake. By Lemma~\ref{appendix: lemma: girth no cycle} there is only one unoccupied vertex the head can move to without losing. The apple is placed on this vertex, forcing the snake to eat the apple on its next move. From there, the head cannot move to its tail, since there was no cycle that contained the snake. It can also not move to the remaining unoccupied vertex since they are not adjacent. It follows that the snake will lose. \qed
\end{proof}

We now return to our scenario where~$G$ is non-Hamiltonian with~$g(G)>6$, the snake grows to length~$|V|-3$, and there is no cycle containing the snake. By Lemma~\ref{appendix: lemma: girth two left}, the apple placer only needs to ensure that when the snake eats the next apple, the remaining unoccupied vertices are not adjacent. This allows us to prove the following.

\setcounter{lemma}{6}
\begin{lemma}\label{appendix: lemma: three left no cycle}
    Let~$G=(V,E)$ be a non-Hamiltonian graph with~$g(G)>6$. Consider the moment the snake grows to length~$\ell = |V|-3$ and suppose there is no cycle in~$G$ that contains the snake. Then, the snake will lose.
\end{lemma}
\begin{proof}
    Let~$S$ be the current snake position with~$s_1$ and~$s_{\ell}$ being the head and tail positions respectively. We will distinguish between two cases. In the first case, the snake is able to visit all three unoccupied vertices without moving to~$S$ in between. In the second case, this is not possible.

    \subsubsection*{Case 1: the snake can visit all vertices in~$\overline{S}$ without moving to~$S$ in between.}
    First, we note that by Lemma~\ref{appendix: lemma: girth no cycle}, the snake can visit the vertices in~$\overline{S}$ in only one order. Let~$u$ be the vertex the snake has to move to first,~$v$ the second, and~$w$ the third. The apple is placed on~$v$. First, the snake has to move to~$u$, after which~$s_{\ell}$ becomes unoccupied. Next, the snake has to eat the apple on~$v$. After eating the apple,~$w$ and~$s_{\ell}$ are the two remaining unoccupied vertices. But~$w$ cannot be adjacent to~$s_{\ell}$ since this would give us the cycle~$(S,u,v,w)$, contradicting that~$G$ is non-Hamiltonian. Hence, the two remaining unoccupied vertices are not adjacent and by Lemma~\ref{appendix: lemma: girth two left}, the snake will lose.

    \subsubsection*{Case 2: the snake cannot visit all vertices in~$\overline{S}$ without moving to~$S$ in between.}
    By Lemma~\ref{appendix: lemma: girth no cycle}, there is at most one vertex the head can move to without guaranteeing a loss. Let this be vertex~$u$. The apple is placed on~$u$, meaning the snake has to eat the apple on its next move. After eating the apple on~$u$, we know the snake cannot move to~$s_{\ell}$, since then the cycle~$(S,u)$ would then have contained~$S$. Thus, one of the remaining unoccupied vertices has to be adjacent to~$u$, otherwise the snake will lose. Let~$v$ be an unoccupied vertex that is adjacent to~$u$ and let~$w$ be the other unoccupied vertex. Since we assumed the snake cannot visit all vertices in~$\overline{S}$ without moving to~$S$ in between, the snake should not be able to move from~$u$ to~$v$ and then to~$w$. It follows that the two remaining unoccupied vertices~$v$ and~$w$ cannot be adjacent and by~\ref{appendix: lemma: girth two left}, the snake will lose \qed
\end{proof}

From our four possible scenarios, we have now shown three will result in a loss for the snake. It remains to show that if there is a cycle of length~$|V|-3$ that contains the snake, then the apple placer has a winning strategy. Note that in this scenario, the snake itself forms is cycle and the head is adjacent to the tail.

\begin{lemma}\label{appendix: lemma: girth head to tail}
    Let~$G=(V, E)$ be a non-Hamiltonian graph with~$g(G)>6$. Consider the moment the snake grows to length~$\ell = |V|-3$ and suppose there is a cycle~$C$ of length~$|V|-3$ that contains the snake. Then, the snake will lose.
\end{lemma}
\begin{proof}
    Since~$C$ has the same length as the snake, the vertices of~$C$ are exactly those that are occupied by the snake. This also means all three unoccupied vertices are in~$\overline{C}$. Note that the three unoccupied vertices cannot form a cycle by themselves, since~$G$ does not contain any cycles of length three. Hence, there is at most one unoccupied vertex that is adjacent to both of the other unoccupied vertices. If such a vertex exists, the apple is placed on it. We will refer to the vertex with the apple as vertex~$a$. 
    
    The snake cannot eat the apple by only moving along~$C$. Hence, at some point, it must move to an unoccupied vertex. By Corollary~\ref{cor: girth visit all}, we know that once the snake leaves~$C$, it cannot return to~$C$ without visiting all unoccupied vertices first. Furthermore, by Observation~\ref{obs: visit to win}, the snake cannot win without visiting all three unoccupied vertices. But since~$a$ is the only vertex that can be adjacent to both other unoccupied vertices, it has to be the second unoccupied vertex the snake visits.

    The scenario when the snake leaves~$C$ is depicted in Figure~\ref{appendix: fig: girth: head to tail}. Let~$u$ be the unoccupied vertex the snake moves to before moving to~$a$, and let~$v$ be the remaining unoccupied vertex. Let~$s_{\ell}$ be the tail position right before the head moves to~$u$. When the head moves to~$u$,~$s_{\ell}$ becomes unoccupied. When the snake eats the apple on~$a$,~$s_{\ell}$ remains unoccupied, and~$v$ and~$s_\ell$ are the two remaining unoccupied vertices. By Lemma~\ref{cor: girth visit all}, the head has to move to~$v$ next, and thus~$v$ must be adjacent to~$a$. But then,~$v$ cannot be adjacent to~$s_\ell$, since this would create the cycle~$(S,u,a,v)$ contradicting that~$G$ is non-Hamiltonian. Hence,~$v$ and~$s_\ell$, the two remaining unoccupied vertices, are not adjacent, and by Lemma~\ref{appendix: lemma: girth two left} the snake will lose. \qed
\end{proof}

  \begin{figure}[t]
\includegraphics[width = 0.7\textwidth]{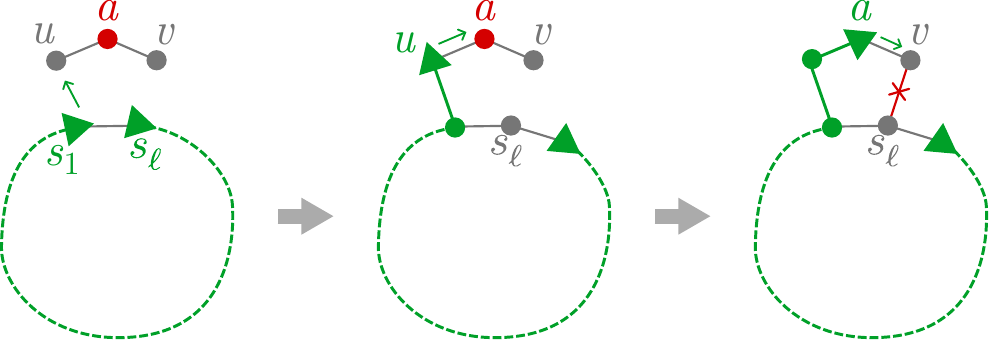}
\centering
\caption{From a cycle of length~$|V|-3$, the snake moves to~$u$ before eating the apple on~$a$. If the snake can move to~$v$, the remaining two unoccupied vertices cannot be adjacent.}
\label{appendix: fig: girth: head to tail}
\end{figure}

\begin{figure}[t]
\includegraphics[width = 6.6cm]{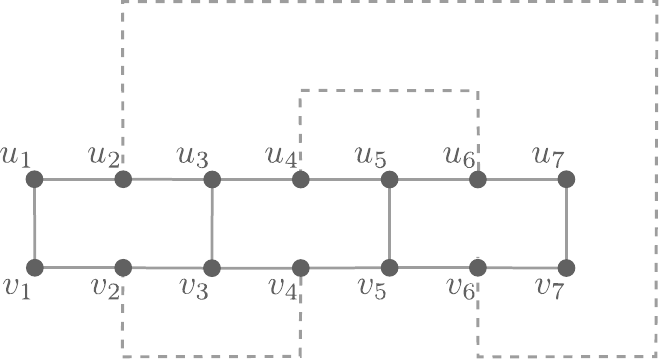}
\centering
\caption{A non-Hamiltonian graph with a girth of 6 that is snake-winnable.}
\label{appendix: fig: girth: partial grid win}
\end{figure} 

\setcounter{obsv}{11}
\begin{obsv}
    The graph in Figure~\ref{appendix: fig: girth: partial grid win} is non-Hamiltonian.
\end{obsv}
\begin{proof}
For any vertex with degree 2, both incident edges must be included in any Hamiltonian cycle. It follows that if the graph in Figure~\ref{appendix: fig: girth: partial grid win} would have a Hamiltonian cycle, then the dashed paths, as well as~$(u_2, u_1, v_1, v_2)$ and~$(u_6, u_7, v_7, v_6)$ must be part of it. Combined, this gives us a path from~$u_4$ to~$v_4$ that contains all vertices except for~$u_3$,~$v_3$,~$u_5$, and~$v_5$. To form a Hamiltonian cycle, we must find a~$u_4 v_4$-path with exactly these remaining four as internal vertices.
But this is not possible, and thus the graph in Figure~\ref{appendix: fig: girth: partial grid win} is not Hamiltonian. \qed
\end{proof}

\begin{obsv}\label{obs: grid two left adjacent}
    Suppose the snake has length~$|V|-2$ and is positioned on some cycle~$C$. Furthermore, suppose the snake covers all vertices on~$C$, and the two remaining unoccupied vertices are adjacent to each other. If both unoccupied vertices are adjacent to some vertex in~$C$, then the snake will win.
\end{obsv}
\begin{proof}
    Let~$u$ and~$v$ be the two unoccupied vertices and suppose they are adjacent to the vertices~$u'$ and~$v'$ on~$C$, respectively. If the next apple is on~$u$, then the snake moves along~$C$ until it reaches~$u'$. From there, it moves to~$u$. It can then move to~$v$ to eat the final apple. Similarly, if the next apple is on~$v$ then the snake moves along~$C$ until it reaches~$v'$.  From there, it moves to~$v$, and then~$u$ to eat the final apple. \qed
\end{proof}

\begin{lemma}\label{appendix: lemma: girth 6 winnable}
    The graph in Figure \ref{appendix: fig: girth: partial grid win} is snake-winnable.
\end{lemma}
\begin{proof}
    In Figure~\ref{appendix: fig: girth: partial grid win cycles} we can see the cycles~$C_1$ and~$C_2$ of~$G$, both have length~$|V|-2$. The snake moves along~$C_1$ or~$C_2$ and changes to the other cycle if the apple is placed outside of its current cycle. Note that while the snake has length at most~$|V|-4$, it can choose which cycle to move along each time it reaches~$u_4$. By employing this strategy, the snake will be positioned on one of the two cycles when it reaches length~$|V|-3$. By symmetry, we may assume this is~$C_1$.

    If the next apple is on~$C_1$, then the snake can eat the apple by moving along~$C_1$. After doing so, it covers the entire cycle, and the two remaining unoccupied vertices are~$u_5$ and~$v_5$. By Observation~\ref{obs: grid two left adjacent}, the snake will win.

    If the next apple is not on~$C_1$, then it is either on~$u_5$ or~$v_5$. If it is on~$v_5$, then the snake moves along~$C_1$ until it reaches~$u_4$. It then first moves to~$u_5$, and then to~$v_5$ to eat the apple. By doing so, the snake will be positioned on~$C_2$, with~$u_3$ and~$v_3$ the remaining unoccupied vertices. By Observation~\ref{obs: grid two left adjacent}, the snake will win.

    If the next apple is on~$u_5$, then the snake moves along~$C_1$ until it reaches~$v_6$. It then first moves to~$v_5$, and then to~$u_5$ to eat the apple. The snake will now be positioned on a new cycle, depicted in Figure~\ref{appendix: fig: girth: even grid new cycle}. The two remaining vertices are~$u_7$ and~$v_7$. By Observation~\ref{obs: grid two left adjacent}, the snake will win. \qed
\end{proof}

\begin{figure}[t]
\includegraphics[width = 0.95\textwidth]{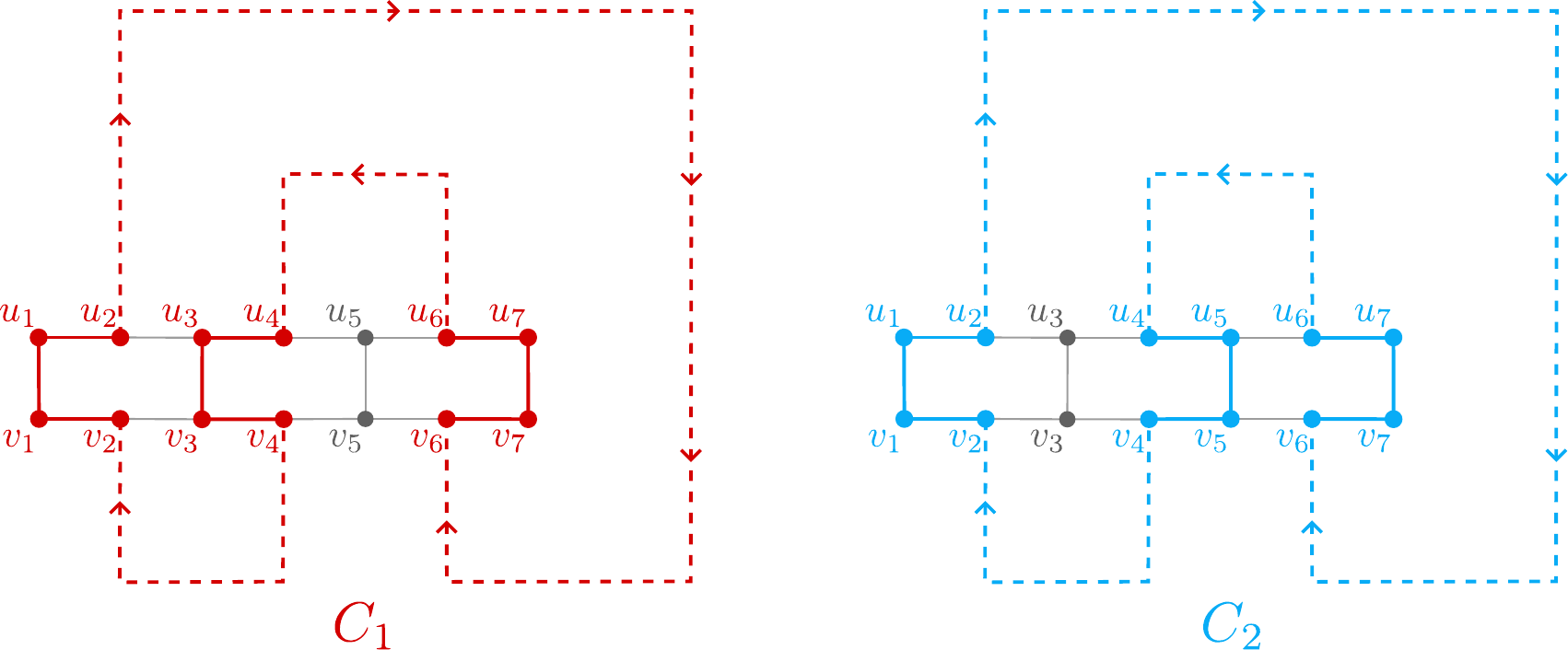}
\centering
\caption{The cycles~$C_1$ and~$C_2$.}
\label{appendix: fig: girth: partial grid win cycles}
\end{figure} 

\begin{figure}[t]
\includegraphics[width = 5.6cm]{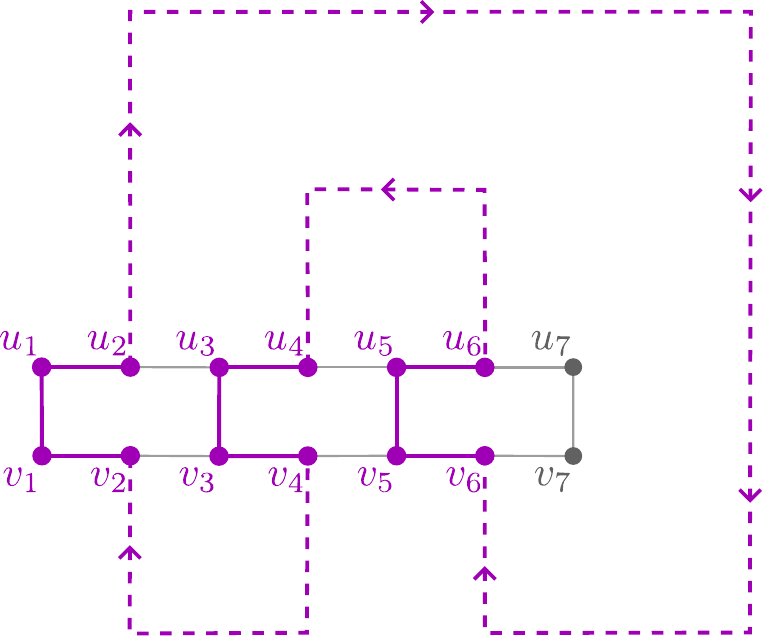}
\centering
\caption{If the apple appears on~$u_5$, the snake moves to a new cycle.}
\label{appendix: fig: girth: even grid new cycle}
\end{figure}

\section{Appendix: Snake-winnable graphs with vertex-connectivity 1}\label{appendix: connectivity}
\setcounter{obsv}{1}
\begin{obsv}
    Let~$v$ be a cut vertex of~$G$. If~$G-v$ has at least 3 components, then~$G$ is not snake-winnable.
\end{obsv}
\begin{proof}
    If~$G-v$ has at least 3 components, then~$G$ cannot have a Hamiltonian path. By Observation \ref{obs: no Ham path},~$G$ is not snake-winnable.
\end{proof}

\setcounter{definition}{8}
\begin{definition}
    Let~$S^t=(s_1^t,\ldots, s_\ell^t)$ be the position of the snake on ~$G$ at time~$t$. The \textbf{head graph} at time~$t$, denoted~$H^t$, is the subgraph of~$G$ induced by~$\overline{S^t}\cup\{s_1^t\}$.
\end{definition}
An example of the head graph can be seen in Figure~\ref{fig: structure: head graph}.

\begin{figure}[t]
\includegraphics[width = 0.33\textwidth]{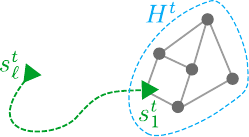}
\centering
\caption{The head graph~$H^t$ is the subgraph induced by~$\overline{S^t}\cup{s_1^t}$.}
\label{fig: structure: head graph}
\end{figure} 

\begin{obsv}\label{appendix: obs: complete is win}
    Suppose at some time~$t$,~$H^t$ is complete. Then, the snake will win.
\end{obsv}
\begin{proof}
    Let~$s_1^t$ be the head of the snake. Since~$H^t$ is complete, the apple~$a$ must be adjacent to~$s_1^t$. The snake can immediately move to~$a$, and for the next head graph, we get~$H^{t+1}=H^t-s_1^t$. Hence, the head graph~$H^{t+1}$ remains complete. The snake can continue to repeat this strategy until it reaches length~$|V|$.
\end{proof}

\begin{figure}[t]
\includegraphics[width = 1\textwidth]{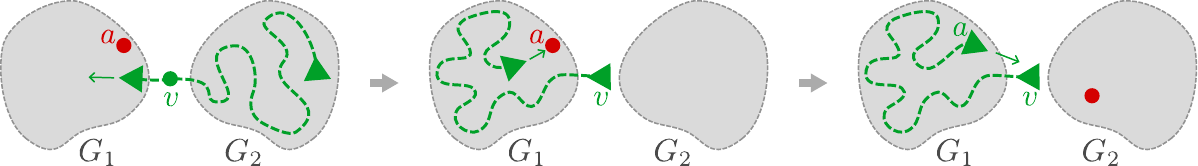}
\centering
\caption{At length~$m$, the snake moves into~$G_1$ until its tail is on~$v$. It then eats the apple on~$a$, after which it occupies every vertex in~$G_1$ and can move to~$v$.}
\label{fig: connectivity: 1c strategy}
\end{figure} 

\setcounter{lemma}{2}
\begin{lemma}
Let~$G=(V, E)$ be a graph with~$\kappa(G)=1$ and let~$v$ be a cut vertex of~$G$. Suppose that~$G_1$ and~$G_2$ are the only two connected components of~$G-v$ with~$|V(G_1)|=|V(G_2)|=m$ and~$m\geq2$. Furthermore, let~$G_1+v$ and~$G_2+v$ both be complete. Then~$G$ is snake-winnable.
\end{lemma}
\begin{proof}
While the snake is shorter than~$m$, it can always move between~$G_1$ and~$G_2$ by moving into the current component until~$v$ becomes unoccupied, and then moving to the other component through~$v$. We consider the moment the snake grows to length~$m$. 

First, suppose the next apple is placed on~$v$. This means~$v$ is currently unoccupied and thus the snake must be entirely in~$G_1$ or entirely in~$G_2$. By symmetry, we may assume the snake is entirely in~$G_1$. Since the snake has length~$m$, it occupies all the vertices in~$G_1$. Furthermore, since~$G_1+v$ is complete the head must be adjacent to~$v$. The snake can eat the apple on~$v$ on the next move, after which the head graph will be~$G_2+v$, which is complete. By Observation~\ref{appendix: obs: complete is win} the snake will win.

Next, suppose after the snake grows to length~$m$, the next apple is not placed on~$v$. By symmetry, we may assume the apple is placed on some vertex~$a$ in~$G_1$, as depicted in Figure~\ref{fig: connectivity: 1c strategy}. If the head is in~$G_2$, then the snake keeps moving into~$G_2$ until~$v$ is unoccupied. It then first moves its head to~$v$, and then into~$G_1$. With its head in~$G_1$ the snake can now repeatedly move to vertices in~$G_1$ that are not~$a$, until its tail is on~$v$. Since the snake has length~$m$, it now occupies all vertices in~$G_1$, except for~$a$. The snake then eats the apple on~$a$ and since~$G_1+v$ is complete, its head on~$a$ will be adjacent to its tail on~$v$. The head moves to~$v$, after which the head graph is~$G_2+v$, which is complete. By Observation~\ref{appendix: obs: complete is win}, the snake will win.
\end{proof}

\begin{lemma}
    Let~$G=(V, E)$ be a graph with~$\kappa(G)=1$ and let~$v$ be a cut vertex of~$G$. Let~$G_1$ and~$G_2$ be two different connected components of~$G-v$, with~$|V(G_1)|=m$. Suppose the head of the snake is in~$G_1$ and the snake has a length of at least~$m+2$. If there is an unoccupied vertex in~$G_2$, then the snake will lose.
\end{lemma}
\begin{proof}
    By Observation~\ref{obs: visit to win}, the snake will have to move to the unoccupied vertex in~$G_2$ at some point in the future. Hence, at some point, the head will have to move from~$G_1$ to~$G_2$. To do so, the head will have to move through~$v$. Thus, the snake must first keep moving into~$G_1$ until it is either entirely in~$G_1$, or only its tail outside of~$G_1$, namely on~$v$. But this is impossible since~$G_1$ only contains~$m$ vertices and the snake has a length of at least~$m+2$.
\end{proof}

\begin{lemma}
    Let~$G=(V, E)$ be a graph with~$\kappa(G)=1$ and let~$v$ be a cut vertex of~$G$. Let~$G_1$ and~$G_2$ be two different connected components of~$G-v$, with~$|V(G_1)|=m$. Suppose the head of the snake is in~$G_1$, the snake has a length of at least~$m+1$, and the apple is on some vertex~$a$ in~$G_1$. If there is some unoccupied vertex in~$G_2$, then the snake will lose.
\end{lemma}
\begin{proof}
    Similar to Lemma~\ref{lemma: too long 1c}, we will show that the head cannot return to~$G_2$. To move from~$G_1$ to~$G_2$, the snake has to move into~$G_1$ until it is either entirely in~$G_1$, or only its tail outside of~$G_1$. Hence, it must occupy at least~$m$ vertices in~$G_1$. But then the head also has to visit~$a$, at which point it will grow to length~$m+2$. By Lemma~\ref{lemma: too long 1c}, the snake will lose.
\end{proof}

It remains to show that if at least one of the two components is incomplete, then the snake will lose. For this, we first distinguish three different types of moves the snake can make, which are depicted in Figure~\ref{fig: structure: move types}. 

In the first type of move, the head moves to the apple and the tail remains in place. This means one vertex is added to the snake and, consequentially, one vertex is removed from the unoccupied set. We will refer to this as a \textit{type~$\alpha$} move. 

In the second type of move, the snake moves to an unoccupied without the apple. In this case, the new head position is added to the snake, and the former tail is removed. This also means one vertex is removed from the unoccupied set and a different vertex is added. We will refer to this as a \textit{type~$\beta$} move. 

In the third move type, the head moves to the tail vertex and the unoccupied set remains unchanged. We will refer to this as a \textit{type~$\gamma$} move.

\begin{figure}[t]
     \centering
     \subfigure[A type~$\alpha$ move.]{\includegraphics[width = 0.44\textwidth]{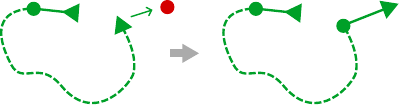}\label{fig: structure: type alpha move}}\hspace{10mm}
     \subfigure[A type~$\beta$ move.]
     {\includegraphics[width = 0.44\textwidth]{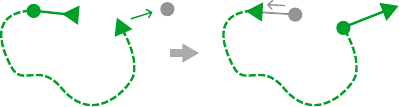}\label{fig: structure: type beta move}}\\
     \vspace{5mm}
     \subfigure[A type~$\gamma$ move.]
     {\includegraphics[width = 0.44\textwidth]{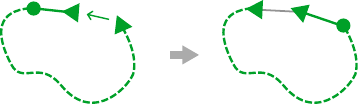}\label{fig: structure: type gamma move}}
        \caption{The three different types of moves.}
        \label{fig: structure: move types}
\end{figure}

\setcounter{obsv}{10}
\begin{obsv}\label{obs: after beta}
    After the snake makes a type~$\beta$ move, there is an unoccupied vertex that is adjacent to the tail. Furthermore, the apple cannot be on this unoccupied vertex.
\end{obsv}
\begin{proof}
    After a type~$\beta$ move, the former tail vertex becomes unoccupied. By Observation~ref{obs: apple on snake}, the apple cannot be on this vertex.
\end{proof}

\begin{lemma} \label{lemma: not comp after beta}
    Let~$G$ be non-Hamiltonian. Then the head graph can never be complete after a type~$\beta$ move.
\end{lemma}
\begin{proof}
    Suppose the head graph would be complete after a type~$\beta$ move. Then by Observation~ref{obs: after beta}, there would be a path from the head to the tail through all of the unoccupied vertices. This contradicts that~$G$ is non-Hamiltonian.
\end{proof}

\begin{corollary}\label{cor: not comp after beta alpha}
    Let~$G$ be non-Hamiltonian and suppose the snake makes a type~$\beta$ move immediately followed by a type~$\alpha$ move. Then after the type~$\alpha$ move, the head graph cannot be complete.    
\end{corollary}
\begin{proof}
    By Observation~ref{obs: after beta}, there must be an unoccupied vertex after the type~$\beta$ move. Since the apple cannot be on this vertex, it remains unoccupied after the type~$\alpha$ move. By the same reasoning as for Lemma~ref{lemma: not comp after beta}, the head graph cannot be complete after the type~$\alpha$ move.
\end{proof}

\begin{definition}
    The \textbf{circumference} of a graph~$G$ is the length of the longest simple cycle in~$G$ and is denoted as is denoted~$\text{circ}(G)$.
\end{definition}

\begin{lemma}\label{lemma: longer than circ}
    Suppose the snake is longer than~$\text{circ}(G)$ and the head graph is not complete. Then the snake will lose.
\end{lemma}
\begin{proof}
    The apple placer uses the following strategy: If possible, the apple is placed on a vertex that is not adjacent to the head. We will show the head graph will remain incomplete for the rest of the game. Hence, the snake either loses before it reaches length~$|V|-1$, or after it reaches length~$|V|-1$, the final apple is not adjacent to the head, causing it to lose.

    Note that the snake can no longer make type~$\gamma$ moves, as it would create a cycle longer than~$\text{circ}(G)$. Suppose the head graph does become complete, and consider the two moves right before this happens. By Lemma~ref{lemma: not comp after beta} and Corollary~ref{cor: not comp after beta alpha}, these must have been two type~$\alpha$ moves. This means that after the the first type~$\alpha$ move, the next apple is placed adjacent to the head. since the apple placer always tries to place the apple on a vertex that is not adjacent to the head, it follows that after the first type~$\alpha$ move, all unoccupied vertices must be adjacent to the head.
    Let~$s_1$ be the head position and~$H$ the head graph after the first type~$\alpha$ move. Then~$H$ is not complete, but~$H-s_1$ must be a complete graph. It follows that there must be some unoccupied  vertex~$v$ that is not adjacent to~$s_1$. This is a contradiction since we previously found that all unoccupied vertices must be adjacent to the head. We conclude that the head graph will remain incomplete.
\end{proof}

\setcounter{lemma}{6}
\begin{lemma}\label{appendix: lemma: 1c not comp}
Let~$G=(V, E)$ be a graph with~$\kappa(G)=1$ and let~$v$ be a cut vertex of~$G$. Let~$G_1$ and~$G_2$ be two different connected components of~$G-v$ and suppose~$G_1+v$ and~$G_2+v$ are not both complete. Then~$G$ is not snake-winnable.
\end{lemma}
\begin{proof}
     If~$G_1$ and~$G_2$ do not have the same number of vertices, then~$G$ is not snake-winnable by Lemma~\ref{lemma: 1c diff size}. Hence, we may assume that~$G_1$ and~$G_2$ have the same number of vertices. Let~$|V(G_1)|=|V(G_2)|=m$. If~$m=1$, then~$G$ has a vertex of degree 1 and by Observation~\ref{obs: degree_one}~$G$ is not snake-winnable. Hence, we may assume that~$m\geq 2$.
     
     Suppose that~$G_1+v$ and~$G_2+v$ are not both complete subgraphs of~$G$. By symmetry, we assume~$G_2+v$ is incomplete.

     When the snake reaches length~$m-1$, there must be some unoccupied vertex~$u_2$ in~$G_2$. The next apple is placed on~$u_2$. When the snake eats the apple on~$u_2$ and grows to length~$m$, there will be at most~$m-2$ vertices in~$G_1$ that are occupied. Hence, the apple placer can place the next apple on some~$u_1$ in~$G_1$, guaranteeing that the head of the snake is in~$G_1$ when it grows to length~$m+1$. Furthermore, there must be an unoccupied vertex in~$G_2$.
     
     First, suppose that there is some unoccupied vertex~$v_1$ in~$G_1$. Then the apple placer places the next apple on~$v_1$. Since the snake has length~$m+1$ and both its head and the apple are in~$G_1$, by Lemma~\ref{lemma: too long apple 1c}, the snake will lose.

     Now suppose that there are no unoccupied vertices in~$G_1$. Recall that we assumed~$G_2+v$ is incomplete. If there is some vertex in~$G_2$ that is not adjacent to~$v$, then the apple placer places the next apple there. This means that if the snake moves to~$v$ and then immediately eats the apple, then after eating the apple, the head graph is incomplete. Since the snake now has length~$m+2$ and the circumference of~$G$ is at most~$m+1$, by Lemma~\ref{lemma: longer than circ}, the snake will lose.

     If the snake first moves to~$v$ and then to a vertex in~$G_2$ without the apple, then a vertex in~$G_1$ will become unoccupied. Since the snake has length~$m+1$ and both its head and the apple are in~$G_2$,  by Lemma~\ref{lemma: too long apple 1c} the snake will lose. \qed
\end{proof}

\end{document}